\documentclass[a4paper,11pt]{article}

\usepackage{amsmath,amssymb}
\usepackage{tabls}
\usepackage[dvipdfmx]{graphicx}
\usepackage{wrapfig}
\usepackage{array}
\usepackage{url}
\usepackage{multirow}
\usepackage{multicol}
\usepackage{here}
\usepackage{theorem}
\usepackage{newproof}

\newtheorem{theorem}{Theorem}
\newtheorem{lemma}{Lemma}

{\theorembodyfont{\normalfont}
  \newtheorem{problem}{Problem}
  
}

\newcommand{\Lmin}[1]{\mathit{LMin}_{#1}}
\newcommand{\Lmax}[1]{\mathit{LMax}_{#1}}
\newcommand{\sort}{\mathrm{sort}}
\newcommand{\Path}{\mathit{path}}

\title{
 Order preserving pattern matching \\ on trees and DAGs
}

\date{}

\author{
  Temma Nakamura \quad
  Shunsuke Inenaga \quad \\
  Hideo Bannai \quad
  Masayuki Takeda \\ \\
  {Department of Informatics, Kyushu University, Japan}\\
  {\small{\texttt{\{temma.nakamura, inenaga, bannai, takeda\}@inf.kyushu-u.ac.jp}}}
}

\begin{document}
\maketitle


\begin{abstract}
  The \emph{order preserving pattern matching} (\emph{OPPM}) problem is,
  given a pattern string $p$ and a text string $t$,
  find all substrings of $t$ which have the same relative orders as $p$.
  In this paper, we consider two variants of the OPPM problem
  where a set of text strings is given as a \emph{tree} or a \emph{DAG}.
  We show that the OPPM problem for a single pattern $p$ of length $m$
  and a text tree $T$ of size $N$ can be solved in $O(m+N)$ time
  if the characters of $p$ are drawn from
  an integer alphabet of polynomial size.
  The time complexity becomes $O(m \log m + N)$
  if the pattern $p$ is over a general ordered alphabet.
  We then show that the OPPM problem for a single pattern
  and a text DAG is NP-complete.
\end{abstract}


\section{Introduction}

The \emph{order preserving pattern matching} (\emph{OPPM}) problem is,
given a pattern string $p$ and a text string $t$,
find all substrings of $t$ which have the same relative orders as $p$.
For instance, let $p = (22, 41, 35, 37)$ and
$t = (63, 18, 48, 29, 42, 56, 25, 51)$.
The relative orders of the characters in $p$ is $1, 4, 2, 3$.
A substring $t[2..5] = (18, 48, 29, 42)$ have the same
relative orders $1, 4, 2, 3$ as $p$,
and hence the occurrence of this substring is reported.
OPPM captures structural isomorphism of strings,
and thus has potential applications in
the analysis of times series such as stock prices, 
and in melody matching of musical sequences~\cite{KimEFHIPPT14,ChoNPS15}.

Let $m$ and $n$ be the lengths of the pattern string $p$ and the text string $t$,
respectively.
Kim et al.~\cite{KimEFHIPPT14} proposed an $O(m \log m + n)$-time
algorithm for the OPPM problem.
Independently, Kubica et al.~\cite{KubicaKRRW13} proposed
an $O(\sort(p) + n)$-time algorithm,
where $\sort(p)$ denotes the time complexity to sort the elements in $p$;
$\sort(p) = O(m \log m)$ for general ordered alphabets
and $\sort(p) = O(m)$ for integer alphabets of size $m^{O(1)}$.
These algorithms are based on the Morris-Pratt algorithm~\cite{MorrisPratt}.
Kubica et al.'s algorithm works when the input strings do not contain
same characters.
Cho et al.~\cite{ChoNPS15} showed how Kubica et al.'s algorithm can be
modified when there are same characters in the input strings,
retaining the same efficiency.
Other types of algorithms for the OPPM problem
have also been proposed~(e.g., see \cite{KimEFHIPPT14,ChoNPS15,ChhabraT16,FaroK16}).

This paper considers two natural extensions to the OPPM problem,
where a set of text strings is given as a tree or a DAG.
We show that the OPPM problem for a single pattern $p$
and a text tree $T$ of size $N$ can be solved in $O(\sort(p)+N)$ time.
(and hence $O(m \log m + N)$ time for general ordered alphabets and
$O(m+N)$ time for integer alphabets polynomial size).
Our method uses a Morris-Pratt type of (non-deterministic)
pattern matching automaton, and 
generalizes the existing results for the OPPM problem
on a single text string.
We then show that the OPPM problem for a single pattern
and a text DAG is NP-complete.

\subsection*{Related work}

The exact pattern matching problem on a single pattern string
and a tree was first considered by Dubiner et al.~\cite{DubinerGM94}.
Their algorithm is based on a (non-deterministic) Morris-Pratt automaton.

Amir and Navarro~\cite{AmirN09} considered
the \emph{parameterized pattern matching} (\emph{PPM}) problem on trees.
Let $\sigma$ be the alphabet size.
They showed that the PPM problem on trees can be solved in
$O(N \log(\min\{\sigma, m\}))$ time, 
provided that the \emph{deterministic} version of
a Morris-Pratt type automaton is available.
However, the size of the deterministic version of such an automaton
can be as large as $O(m^2)$.
Hence, their algorithm takes $O(m^2 + N(\min\{\sigma, m\}))$
time in the worst case\footnote{Simon~\cite{Simon94} proposed
  an $O(m)$-space Morris-Pratt automaton
  for exact pattern matching, however, it is unclear if this can be extended to PPM or OPPM.}.

Recall that the running time per text character of
a \emph{non-deterministic} Morris-Pratt automaton depends
on the number of failure transitions used per text character.
The key analysis of a total linear running time
of this method on a single text string
is that this number is amortized constant.
The same amortization argument holds for its OPPM and PPM variants
on a single text string.

The difficulty in using a non-deterministic Morris-Pratt automaton for
a tree text is that if we simply run the automaton on the tree as is,
then the above amortization argument does not hold.
It seems that this point was overlooked even in the
exact pattern matching problem on trees (see the proof of Lemma 2.2
of the work by Dubiner et al.~\cite{DubinerGM94}).

Still, we will show that a small trick permits us to bound the
number of failure transitions per character to amortized constant,
achieving our result for the OPPM problem on trees.
We here emphasize that the same trick can be employed in
\emph{any} variant of a non-deterministic Morris-Pratt type automaton.
This implies that it is actually possible to
solve the exact pattern matching problem on trees in $O(m+N)$ time
using the Morris-Pratt automaton,
and the PPM problem on trees in $O((m+N) \log(\min\{\sigma, m\}))$ time.
Both of these results are optimal;
the former is clear, and the latter matches the lower bound of the PPM problem
in the comparison model~\cite{AmirN09}.

Several results for the exact pattern matching problem
with a single pattern string and a labeled graph are known
(e.g.~\cite{ParkK95,AmirLL00}).
See a survey~\cite{AmirLL14} for other schemes of pattern matching
on graph texts.

Amir and Navarro~\cite{AmirN09} showed 
the PPM problem on DAGs is NP-complete.
Coupled with their afore-mentioned results on trees,
we can observe that
the PPM and OPPM problems have similar complexities on trees and DAGs.


\section{Preliminaries}

Let $\Sigma$ be a totally ordered alphabet.
An element of $\Sigma^*$ is called a {\em string}.
The length of a string $w$ is denoted by $|w|$.
The empty string $\varepsilon$ is a string of length~$0$.
For a string $w = x y z$, $x$, $y$ and $z$ are called
a \emph{prefix}, \emph{substring}, and \emph{suffix} of $w$, respectively.
The length of a string $w$ is denoted by $|w|$.
The $i$-th character of a string $w$ is denoted by $w[i]$ for each $1 \leq i \leq |w|$.
For a string $w$ and two integers $1 \leq i \leq j \leq |w|$,
let $w[i..j]$ denote the substring of $w$ that begins at position $i$ and ends at
position $j$. For convenience, let $w[i..j] = \varepsilon$ when $i > j$.

Any strings $x, y \in \Sigma^*$ of equal length $m$
are said to be \emph{order-isomorphic}~\cite{KubicaKRRW13} if
the relative orders of the characters of $x$ and $y$ are the same,
i.e., $x[i] \le x[j] \iff y[i] \le y[j]$ for any $1 \leq i,j \leq m$.
A non-empty pattern string $p$ is said to \emph{order-preserving match}
(\emph{op-match} in short)
a non-empty text string $t$ iff
there is a position $i$ in $t$ such that $p \approx t[i-|p|+1..i]$.
The \emph{order-preserving pattern matching} (\emph{OPPM}) problem
is to find all such text positions.

For any string $x$ of length $m$,
an integer $i$~($1 \leq i < m$) is said to be
an \emph{order-preserving border} of $x$ if $x[1..i] \approx x[m-i+1..m]$.

We consider the following two variants of the OPPM problem:
Assume that the set of text strings is given
as a tree $T$ or a DAG $G$ where each edge is labeled by
a character from $\Sigma$.
A pattern string $p$ of length $m$ is said to op-match
a tree $T$ (resp. a DAG $G$) if $p$ op-matches the label of a path in $T$
(resp. $G$).
In this paper, we consider
the locating version of the OPPM on trees
and the decision version of the OPPM on DAGs,
which are respectively defined as follows.

\begin{problem}[The OPPM problem on trees] \label{prob:opmatch_tree}
  Given a pattern string $p$ and
  an edge-labeled tree $T$,
  report the final node of every path in $T$ that $p$ op-matches. 
\end{problem}

\begin{problem}[The OPPM problem on DAGs] \label{prob:opmatch_dag}
  Given a pattern string $p$ and
  an edge-labeled DAG $G$,
  determine whether $p$ op-matches $G$ or not.
\end{problem}


\section{Order preserving pattern matching on trees}

Our algorithm for order preserving pattern matching on
a text tree is inspired by the algorithms
for order preserving pattern matching on a text string~\cite{KubicaKRRW13,ChoNPS15}.
We will utilize the following tools in our algorithm.

For any string $x$ let $\Lmax{x}$ be an array of length $|x|$ such that
$\Lmax{x}[i] = j$ if $x[j] = \max\{x[k] \mid 1 \leq k < i, x[k] \leq x[i]\}$.
Similarly, let $\Lmin{x}$ be an array of length $|x|$ such that
$\Lmin{x}[i] = j$ if $x[j] = \min\{x[k] \mid 1 \leq k < i, x[k] \geq x[i]\}$.
If there is no such $j$, then let $\Lmax{x}[i] = 0$ and $\Lmin{x}[i] = 0$,
respectively.
If there are several such $j$'s,
then we select the rightmost one among them.

\begin{lemma}[\cite{KubicaKRRW13}] \label{lem:Lmin_Lmax_sort}
  Given a string $x$,
  we can compute the $\Lmax{x}$ and $\Lmin{x}$ arrays in $O(\sort(x))$ time,
  where $\sort(x)$ is the time to sort the elements of $x$.
\end{lemma}

\begin{lemma}[\cite{ChoNPS15}] \label{lem:Lmin_Lmax}
  For strings $x$ and $y$, assume $x[1..i] \approx y[1..i]$ for
  $1 \leq i < \min\{|x|, |y|\}$.
  Let $a = \Lmax{x}[i+1]$ and $b = \Lmin{x}[i+1]$.
  Let $\alpha$ be the condition that $y[a] < y[i+1]$
  and $\beta$ be the condition that $y[i+1] < y[b]$.
  Then, $x[1..i+1] \approx y[1..i+1] \Longleftrightarrow
  (\alpha \land \beta) \lor (\lnot \alpha \land \lnot \beta)$.
  In case $a$ or $b$ is equal to $0$,
  we assume the respective condition $\alpha$ or $\beta$ is true.
\end{lemma}

Let $p$ be a pattern string of length $m$.
We compute the \emph{order-preserving border array} $B_p$ of length $m$ such that $B_p[1] = 0$ and $B_p[i] = \max\{j \mid j < i, p[1..j] \approx p[i-j+1..i]\}$ for $2 \leq i \leq m$. Namely, $B_p[i]$ stores the largest order-preserving border of the prefix $p[1..i]$.
Suppose that $\Lmin{p}$ and $\Lmax{p}$ have already been computed
using Lemma~\ref{lem:Lmin_Lmax_sort}.
Kubica et al.~\cite{KubicaKRRW13} showed that 
using a variant of the Morris-Pratt (MP) algorithm~\cite{MorrisPratt}
based on Lemma~\ref{lem:Lmin_Lmax},
the $B_p$ array can be computed in $O(m)$ time.
Then, given a text string $t$ of length $n$,
all positions $i$ in $t$ where $p \approx t[i-m+1..i]$
can be computed in $O(n)$ time.

We will extend the above algorithm to the case where
the text strings are given as a tree $T$ of size $N$.
It is convenient to consider an MP-style automaton 
$\mathcal{A}_p$ based on the op border array $B_p$ such that
the set of states is $\{s_0, \ldots, s_m\}$;
the initial state is $s_0$;
the only accepting state is $s_m$;
for each $1 \leq i \leq m$ there is a goto transition from $s_{i-1}$ to $s_{i}$ with character $c = p[i]$;
and there is a failure transition from $s_{i}$ to $s_{j}$
iff $B_p[i] = j$.
See Fig.~\ref{fig:op-automaton} for a concrete example of $\mathcal{A}_p$.
We run $\mathcal{A}_p$ over the text tree $T$ in depth first manner.
Let $v$ be any node in $T$.
For any $1 \leq i \leq m$,
let $v_i$ denote the $i$th ancestor of $v$ (if it exists),
and $\Path(v_i,v)$ the path label from $v_i$ to $v$.
At each node $v$ visited during the DFS,
we compute the length $\ell(v)$ of the longest path $v_{\ell(v)}, \ldots, v$
such that $p[1..\ell(v)] \approx \Path(v_{\ell(v)}, v)$.
We report every node $v$ with $\ell(v) = m$.
If $\ell(v) < m$,
then we store a pointer to state $s_{\ell(v)}$ at node $v$,
and otherwise we store a pointer to state $s'$ at node $v$,
where $s'$ is the state pointed by the failure transition of $s_{\ell(v)}$.
\begin{figure}[t] 
 \centerline{
 \includegraphics[scale=0.7]{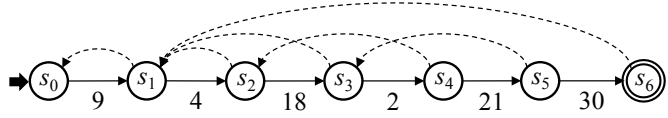}
 }
 \caption{The MP-style automaton $\mathcal{A}_p$ for pattern string $p = (9, 4, 18, 2, 21, 30)$.
   The solid arcs denote the goto transitions, and the broken arcs do the failure transitions. }
\label{fig:op-automaton}
\end{figure}

Suppose we have just visited node $v$.
Initially, let $\ell \leftarrow \ell(v)$.
Let $u$ be any child of $v$ and let $c$ be the edge label from $v$ to $u$.
We proceed to node $u$ and find $\ell(u)$.
We test if the characters $\Path(v_\ell, u)[a]$ and $\Path(v_\ell, u)[b]$
satisfy one of the conditions in Lemma~\ref{lem:Lmin_Lmax},
where $a = \Lmax{p}[\ell+1]$ and $b = \Lmin{p}[\ell+1]$.
If they do, then we let $\ell(u) = \ell+1$ and proceed with the DFS.
Otherwise, then let $\ell \leftarrow B_p[\ell]$,
and repeat the above procedure
until we find the largest $\ell$ with which one of the conditions
in Lemma~\ref{lem:Lmin_Lmax} is satisfied.
For each candidate $\ell$ above,
accessing the character $\Path(v_{\ell}, u)[a]$ from the currently visited node
$u$ means accessing the $(\ell-a+1)$th ancestor of $u$.
Let $L$ be the length of the longest path in $T$.
During the DFS, we store the edge labels of the current path from the root
into an array of length $L$.
Using this array we can access $\Path(v_{\ell}, u)[a]$ (and $\Path(v_{\ell}, u)[b]$)
in $O(1)$ time.
It is easy to update this array during the DFS, in total $O(N)$ time.
When we come back to node $v$ after a back track,
then we resume pattern matching from state $s_{\ell(v)}$ of $\mathcal{A}_p$
using a pointer stored at $v$,
and proceed to the next child of $v$.
This pointer is used after a back track.

One delicacy remains.
For a single text string
the number of candidate $\ell$'s,
which is the same as the number of failure transitions used per
text character, can be amortized constant.
This amortization argument is based on the fact that
the total number of times the failure transitions are used
for the whole text
cannot exceed the total number of times the goto transitions are used
in the automaton $\mathcal{A}_p$, which is bounded by the length
of the single text string.
However, in our tree case, this amortization argument does not hold
if we carelessly continue the DFS at branching nodes
that are close to leaves,
leading to $O(mN)$ worst case time.
See Appendix for an example.
To avoid this, at each node $u$ of the tree $T$
we store the distance $D_u$ between $u$ and a furthest leaf
in the subtree rooted at $u$.
Namely, $D_u$ is the length of the longest path from $u$ and a leaf below $u$.
Suppose that we are currently visiting a node $u$ during the DFS
with $D_u \geq m-\ell$,
and that the respective state of the automaton $\mathcal{A}_p$ is $s_\ell$
(Notice that if $D_u < m-\ell$, then clearly the pattern $p$ does not op-match
any path ending in the subtree under $u$,
and thus we need not search the subtree under $u$ in this case).
Let $v$ be any child of $u$.
If at least one of the conditions of Lemma~\ref{lem:Lmin_Lmax} is satisfied,
then we let $\ell \leftarrow \ell + 1$ and the DFS proceeds to $v$.
Otherwise, we let $\ell \leftarrow B_p[\ell]$
and check if $D_u \geq m-\ell$ holds each time the value of $\ell$ gets updated.
We stop updating $\ell$ as soon as
we encounter $\ell$ for which $D_u < m-\ell$,
and the DFS immediately starts a back track from this child $v$.
This permits us to charge the cost for amortization to
the length $D_u$ of this longest path under $u$.
Thus, this method correctly finds all locations in the tree $T$ where
$p$ op-matches.
We can easily precompute $D_u$ for all nodes $u$ in $T$
in $O(N)$ total time by a standard traversal on $T$.

Overall, we obtain the main result of this section.
\begin{theorem}
  Given a pattern $p$ of length $m$ and a text tree $T$ of size $N$,
  the OPPM problem on trees (Problem~\ref{prob:opmatch_tree})
  can be solved in $O(\sort(p) + N)$ time.
\end{theorem}
For general ordered alphabets, $\sort(p) = O(m \log m)$.
For integer alphabets of size $m^{O(1)}$, $\sort(p) = O(m)$.
For integer alphabets of size $N^{O(1)}$, $\sort(p) = O(m+N)$.


\section{Order preserving pattern matching on DAGs}

A string $x$ is said to be a \emph{subsequence} of
another string $t$ if there exists an increasing
sequence of positions $1 \leq i_1 < \cdots < i_{|x|} \leq |t|$
of $t$ such that $x = t[i_1] \cdots t[i_{|x|}]$.
Intuitively, $x$ is a subsequence of $t$
if $x$ can be obtained by removing zero or more characters from $t$.

The \emph{order-preserving subsequence matching problem}
(\emph{OPSM} in short) is,
given a pattern string $p$ and a text string $t$,
to determine whether there is a subsequence $s$ of $t$
such that $p \approx s$.
This problem is known to be NP-complete~\cite{BoseBL98}.

\begin{theorem}
  The OPPM problem on DAGs (Problem~\ref{prob:opmatch_dag}) is NP-complete.
\end{theorem}

\begin{proof}
  It is clear that the OPPM problem on DAGs is in NP.
  The proof for NP-completeness is via the above OPSM problem.
  Suppose $p$ is a given pattern string
  and $t$ is a given text string for the OPSM problem.
  Consider the \emph{directed acyclic subsequence graph}
  (\emph{DASG} in short)~\cite{Baeza-Yates91} $G_t = (V, E)$
  such that
  \begin{eqnarray*}
    V & = & \{v_0, \ldots, v_{|t|}\}, \\
    E & = & \{(v_i, c, v_j) \mid c = t[j] \mbox{ and } t[k]\neq c \mbox{ for } i < \forall k < j \}.
  \end{eqnarray*}
  The DASG $G_t$ represents all subsequences of $t$,
  i.e., $s$ is a subsequence of $t$ if and only if
  there is a path in $G_t$ of which label coincides with $s$
  (see Fig.~\ref{fig:dasg} for an example).
  Hence, if we can solve the op-matching problem
  for the given pattern string $p$ and the DASG $G_t$,
  then we can immediately solve the OPSM problem.
  The size of DASG $G_t$ is clearly polynomial
  in the length of the given text $t$
  and $G_t$ can be easily constructed in polynomial time.
  This completes the proof.
\end{proof}

\begin{figure}[tb] 
 \centerline{
 \includegraphics[scale=0.6]{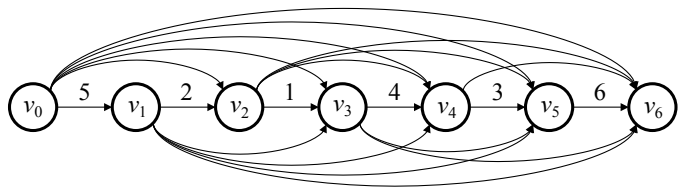}
 }
 \caption{The DASG $G_t$ of string $t = (5, 2, 1, 4, 3, 6)$. At each node, every in-coming edge is labeled with the same character.}
\label{fig:dasg}
\end{figure}


\bibliographystyle{abbrv}
\bibliography{ref}

\clearpage

\appendix

\section{Appendix}

In this appendix,
we show that applying a Morris-Pratt type algorithm
for the order-preserving pattern matching (OPPM)
problem to a text tree as is can lead to $O(mN)$ worst-case time
for some instances.

Consider a complete binary tree $T$ of height $h$ such that
\begin{itemize}
\item
  the path from the root to each node of depth $h-2$
  is labeled by non-negative integers in increasing order,
\item
  the two out-going edges of each node of depth $h-2$
  are labeled by $0$ and $1$, and
\item
  the edge leading to any leaf is labeled by any character.
\end{itemize}
Also, consider pattern string $p = (2, 3, \ldots, m+1)$ of length $m \leq h-2$.
See Fig.~\ref{fig:bad-example} for a concrete example of tree $T$ and pattern $p$.

\begin{figure}[tbhp] 
 \centerline{
   \includegraphics[scale=0.45]{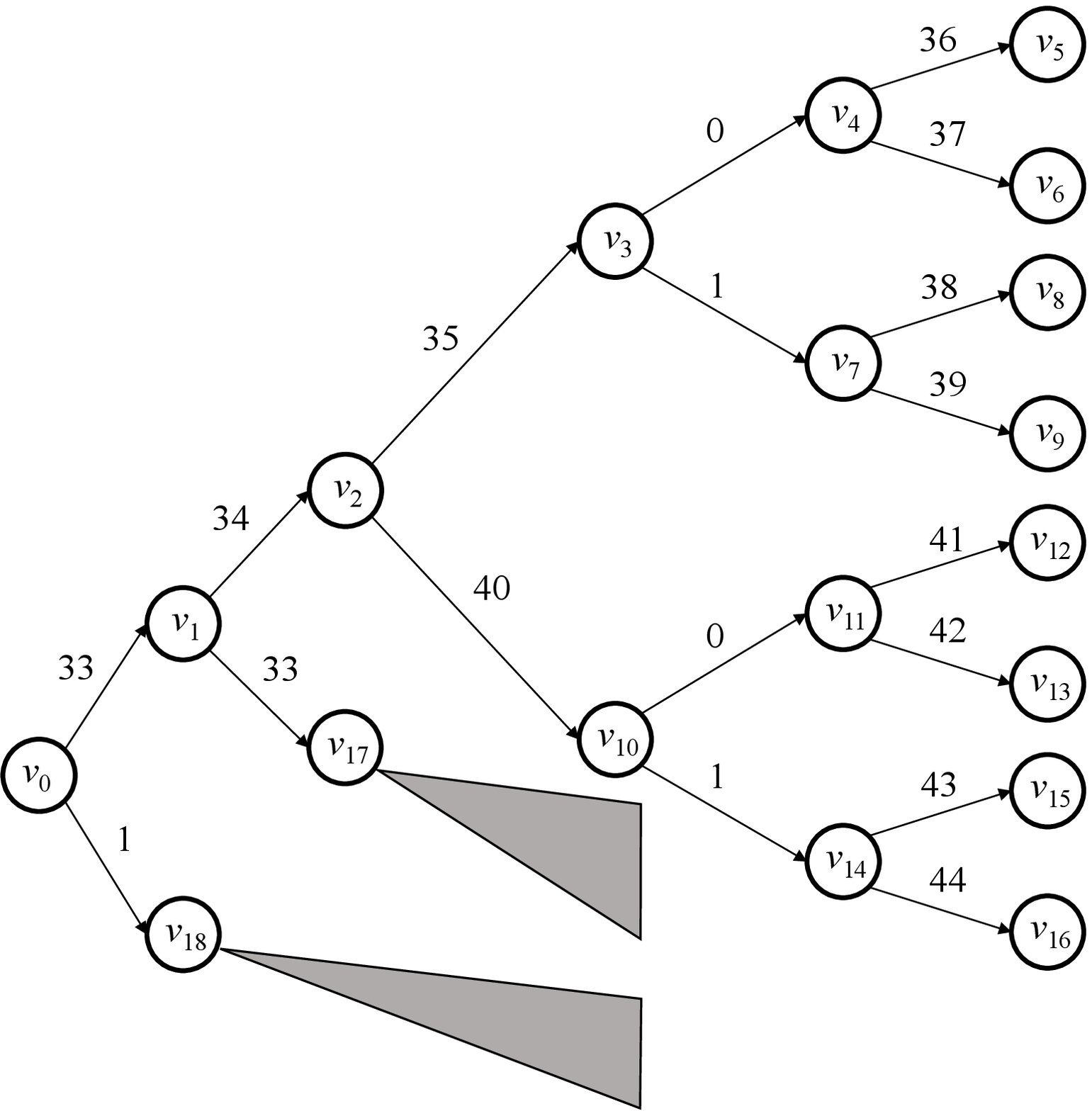}
   \hfill
   \includegraphics[scale=0.6]{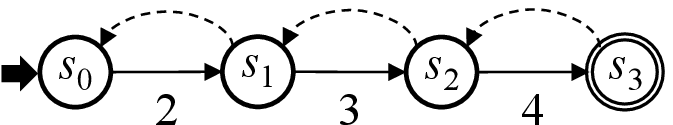}
 }
 \caption{
   To the left is a tree complete binary tree $T$ of height $h = 6$,
   where the gray subtrees are omitted.
   To the right is the MP-style automaton $\mathcal{A}_p$ for pattern string $p = (2, 3, 4)$ with $m = 3$.
   Suppose we have arrived at node $v_3$ during the DFS.
   Since $p = (2, 3, 4) \approx (33, 34, 35)$,
   the current state of $\mathcal{A}_p$ is $s_3$.
   We proceed to node $v_4$ with the DFS.  
   Since there was a match at the previous node $v_3$,
   we follow the failure transition of $s_3$,
   the current state of $\mathcal{A}_p$ becomes $s_2$,
   and a pointer to $s_2$ is stored at node $v_3$.
   Now we perform order preserving pattern matching.
   Since $(2, 3, 4) \not \approx (34, 35, 0)$,
   we follow the failure transition of $s_2$ and
   the current state of $\mathcal{A}_p$ becomes $s_1$.
   Since $(2, 3) \not \approx (35, 0)$,
   we again follow the failure transition of $s_1$ and
   the current state of $\mathcal{A}_p$ becomes $s_0$.
   Finally, since $(2) \approx (1)$,
   the current state of $\mathcal{A}_p$ becomes $s_1$
   and a pointer to $s_1$ is stored at node $v_4$.
   When we come back to node $v_3$ after a back track,
   the same amount of work as above is needed for the other child $v_7$.
   This applies to any node of the tree $T$
   with out-going edges labeled by $0$ and $1$.
 }
\label{fig:bad-example}
\end{figure}

The failure transition of each state $s_{i}$ of the automaton $\mathcal{A}_p$
points to state $s_{i-1}$.
Suppose we have arrived at a node $v$ of depth $h-2$ during the DFS.
By the definitions of tree labels and $p$,
pattern $p$ op-matches the path of length $m$ ending at this node $v$,
and hence the current state of the automaton is $s_{m}$.
Now we proceed to a child of $v$, say $u$,
and suppose that the edge label from $v$ to $u$ is $0$
(the case with the edge label being $1$ is analogous).
Since the edge label is $0$,
the longest prefix of $p$ that op-matches a path ending at $u$ is $p[1] = 1$
and hence $m-1$ failure transitions are used at this point of the DFS.

The tree $T$ contains $K = 2^h$ leaves and hence
contains $N = 2K-1$ total nodes.
The number of tree edges labeled by $0$ is $K/2 = (N+1)/4$.
Thus, the total number of times the failure transitions are used
is $\Theta(mN)$ for this tree $T$ and pattern $p$.

We remark that the same instance leads to an $O(mN)$ worst-case time
bound for the parameterized pattern matching (PPM) problem on trees as well.

For the exact pattern matching problem,
if we allow out-going edges of a node to have the same labels,
then it is easy to show $O(mN)$ worst-case time bound:
In the above example, 
replace $0$ with $b$,
$1$ with $c$, and 
all the other tree edge labels with $a$.
For pattern $p = a^m$, 
the total number of failure transitions of the Morris-Pratt automaton
used for this new tree and $p$ is $\Theta(mN)$.

\end{document}